\documentclass[12pt]{article}
\usepackage{amssymb}
\usepackage{times}
\usepackage{amsmath,amsfonts,amsthm,mathtools}
\usepackage{graphicx}
\usepackage{textcomp}
\usepackage{xcolor}
\usepackage{multirow}
\usepackage{url}
\usepackage{xcolor}
\usepackage{geometry}
\usepackage{hyperref}
\usepackage[T1]{fontenc}

\newtheorem{theorem}{Theorem}
\newtheorem{corollary}{Corollary}
\newtheorem{claim}{Claim}

\newcommand{\arc}[2]{(#1,#2)}
\newcommand{\edge}[2]{\{#1,#2\}}
\newcommand{\inneighbor}[2]{N^-_{#1}{(#2)}}
\newcommand{\outneighbor}[2]{N^+_{#1}{(#2)}}
\newcommand{\outneighborc}[2]{N^+_{#1}{[#2]}} %closed inneighbors
\newcommand{\prob}[1]{{\sc{#1}}}
\newcommand{\cut}[3]{{{\textsf{cut}}}_{#3}(#1, #2)}
\newcommand{\nph}{{\sf{NP}}-hard}
\newcommand{\apxh}{{\sf{APX}}-hard}
\newcommand{\nphns}{{\sf{NP}}-hardness}

\newcommand{\mprob}{{\sc{Energy-Saving Partition of DAG}}} %% name of the main problem
\newcommand{\mprobshort}{{\sc{ESP-DAG}}}
\newcommand{\probva}{{\sc{Size Bounded Energy-Saving Bipartition of DAG}}}
\newcommand{\probvashort}{{\sc{SB-ESBP-DAG}}}
\newcommand{\wa}{{\sf{W[1]}}}
\newcommand{\wah}{{\sf{W[1]}}-hard}
\newcommand{\wahns}{{\sf{W[1]}}-hardness}
\newcommand{\transt}[3]{{\textsf{Tran}_{#3}(#1, #2)}}
\newcommand{\btranst}[3]{{\textsf{BTran}_{#3}(#1, #2)}}
\newcommand{\bigo}[1]{O(#1)}
\newcommand{\fpt}{{\sf{FPT}}}
\newcommand{\yes}{Yes}
\newcommand{\yesins}{{\yes}-instance}
\newcommand{\abs}[1]{|#1|}
\newcommand{\setmid}{:}
\newcommand{\babs}[1]{\left|#1\right|}

\newcommand{\probdef}[3]
{
\begin{center}
\begin{tabular}{|l p{0.85\textwidth}|}\hline
\multicolumn{2}{|l|}{\sc{#1}} \\ \hline
{\bf{Input:}}    & #2 \\ 
{\bf{Question:}} & #3 \\ \hline
\end{tabular}
\end{center}
}

\usepackage{geometry}
\begin{document}

\title{On the Complexity of Minimizing Energy Consumption of Partitioning DAG Tasks}

\author{Wei Liu$^1$, Jian-jia Chen$^1$, Yongjie Yang$^2$}
%\author[lc]{Jian-Jia Chen}
%\author[yyj]{Yongjie Yang}
%\affiliation[lc]{organization={Department of Computer Science, TU Dortmund University},
     %       city={Dortmund},
     %       postcode={44227},           
    %        country={Germany}}
%\affiliation[yyj]{organization={Chair of Economic Theory, Saarland University},
    %         city={Saarbr\"{u}cken},
    %         postcode={66123},
    %         country={Germany}}
\date{\small{$^1$Department of Computer Science, TU Dortmund University, Dortmund, 44227, Germany\\
{\textsf{liuwei-neu@hotmail.com, jian-jia.chen@cs.tu-dortmund.de}}\\[3mm] 
$^2$Chair of Economic Theory, Saarland University, Saarbr\"{u}cken, 66123, Germany\\
{\textsf{yyongjiecs@gmail.com}}}}

\maketitle

\begin{abstract}
We study a graph partition problem where the input is a directed acyclic graph (DAG) representing tasks as vertices and dependencies between tasks as arcs. The goal is to assign the tasks to~$k$ heterogeneous machines in a way that minimizes the total energy consumed for completing the tasks.  We first show that the problem is {\nph}. Then, we present polynomial-time algorithms for two special cases: one where there are only two machines, and another where the input DAG is a directed path. Finally, we examine a variant where there are only two machines, with one capable of executing a limited number of tasks, and demonstrate that this special case remains computationally hard.
%\smallskip
%
%\noindent{\bf{Keywords:}} graph partition problems; NP-hardness;  DAG; task allocation problems 
\end{abstract}

%% \linenumbers
\section{Introduction}
Tasks that are represented by directed acyclic graphs (DAG tasks) are ubiquitous in many applications, including for instance cloud computing, deep neural network, workflow scheduling, etc.~\cite{DBLP:conf/icpp/Duan021,DBLP:journals/csur/KwokA99,DBLP:conf/lcn/LiLLXJ21,DBLP:journals/tc/LinSUGRC23,DBLP:journals/tjs/RajakKPRD23,DBLP:conf/iclr/Thost021}. In this paper, we investigate the complexity of a new graph partition problem which models the scenario where DAG tasks are deemed to be assigned to~$k$ heterogeneous machines (e.g., execution units in distributed systems, clusters of cores in heterogeneous multicore systems, etc.), with the objective to minimize the energy consumption for the computation of these tasks under natural restrictions. More precisely, in this problem, we are given a DAG whose vertices represent tasks, in which the energy consumption of a task depends on which machine is allocated for its execution. An arc from a task~$a$ to a task~$b$ means that the computation of~$b$ requires the output of task~$a$. So, if~$a$ and~$b$ are assigned to different machines, the output of~$a$ needs to be transferred to the machine executing the task~$b$, which also incurs energy consumption. We note that when there are multiple outneighbors of a task~$a$ assigned to a machine~$i$ different from that of~$a$, we need to transfer the output of~$a$ to the machine~$i$ only once. 

We assume that the energy consumption associated with data transfer depends solely on the volume of transferred data, and that the impact of other factors is negligible. Under this assumption, energy consumption for data transfer can be modeled as a univariate function that maps tasks to numerical values. In a more general scenario, energy consumption could also consider additional factors such as the type of data being transferred, the identities of the machines involved, and other relevant parameters. Nevertheless, we prove that even in our simplified case, the problem remains {\nph}. 

Now we formulate the problems. We assume the reader is familiar with the basics in graph theory~\cite{DBLP:books/sp/BG2018} and parameterized complexity theory~\cite{DBLP:books/sp/CyganFKLMPPS15,DBLP:series/txcs/DowneyF13,DBLP:journals/siamcomp/DowneyF95}. 

A {\emph{directed graph}} ({\emph{digraph}}) is a tuple $G=(V, A)$ where~$V$ is a set of vertices and~$A$ is a set of arcs. Each {\emph{arc}} is defined as an ordered pair of vertices. An arc from a vertex~$v$ to a vertex~$u$ is denoted by~$\arc{v}{u}$. We say that the arc $\arc{v}{u}$ {\emph{leaves}}~$v$ and {\emph{enters}}~$u$. We also use~$A(G)$ to denote the set of arcs of~$G$, and use~$V(G)$ to denote the set of vertices of~$G$. The set of {\emph{outneighbors}} (respectively, {\emph{inneighbors}}) of a vertex~$v\in V(G)$ is defined as $\outneighbor{G}{v}=\{u\in V(G) \setmid \arc{v}{u}\in A(G)\}$ (respectively, $\inneighbor{G}{v}=\{u\in V(G) \setmid \arc{u}{v}\in A(G)\}$). The set of {\emph{closed outneighbors}} of~$v\in V(G)$ is defined as $\outneighborc{G}{v}=\outneighbor{G}{v} \cup \{v\}$. Vertices  without any outneighbors are called {\emph{sinks}}, and those without inneighbors are called {\emph{sources}} of~$G$. A {\emph{DAG}} is a digraph without directed cycles. 
For a vertex $v\in V(G)$ and a subset $S\subseteq V(G)$,  we define
\begin{equation*}
    {\bf{1}}_G(v, S)=
    \begin{cases}
        1, & \text{if }\outneighbor{G}{v}\cap S\neq \emptyset,\\
        0, & \text{otherwise},\\
    \end{cases}
\end{equation*}
which indicates whether~$v$ has at least one outneighbor from~$S$ in~$G$.

For a function $f: X\rightarrow Y$ and an element $y\in Y$, we use~$f^{-1}(y)$ to denote the set consisting of all $x\in X$ such that $f(x)=y$. 

For an integer~$k$, let~$[k]$ denote the set of all positive integers less than or equal to~$k$. We also use~$[k]$ to represent a set of~$k$ machines.   
Let $f: V(G)\rightarrow [k]$ be a function that assigns the vertices of~$G$ to the~$k$ machines. Additionally, 
let $p: V(G)\times [k]\rightarrow \mathbb{R}_{\geq 0}$ be a function where $p(v, i)$ defines the energy consumed for computing task~$v$ on machine~$i$, for each $v\in V(G)$ and $i\in [k]$. Furthermore, let $q: V(G)\rightarrow \mathbb{R}_{\geq 0}$ be a function that specifies the energy consumption for transferring the outputs of the vertices in~$G$. The total energy consumption required to transfer all necessary data under the assignment~$f$ is then defined as: 
\begin{equation}
\label{eq-transmiting-time}
    \transt{f}{q}{G}=\sum_{i\in [k]} \sum_{v\in f^{-1}(i)} \sum_{j\in [k]\setminus \{i\}}  q(v) \cdot {\bf{1}}_G(v, f^{-1}(j)).
\end{equation}

The {\mprob} (\mprobshort) problem is defined as follows. 

\probdef
{\mprobshort}
{A DAG~$G$, two functions $p: V(G)\times [k]\rightarrow \mathbb{R}_{\geq 0}$ and $q: V(G)\rightarrow \mathbb{R}_{\geq 0}$, a rational number~$r$.}
{Is there an assignment $f: V(G)\rightarrow [k]$ so that 
\[\left(\sum_{v\in V(G)} p(v, f(v))\right) + \transt{f}{q}{G}\leq r?\]}

See Figure~\ref{fig-DAG-illustraion} for an illustration of the problem.

\begin{figure}[ht!]
    \centering
    \includegraphics[width=0.95\textwidth]{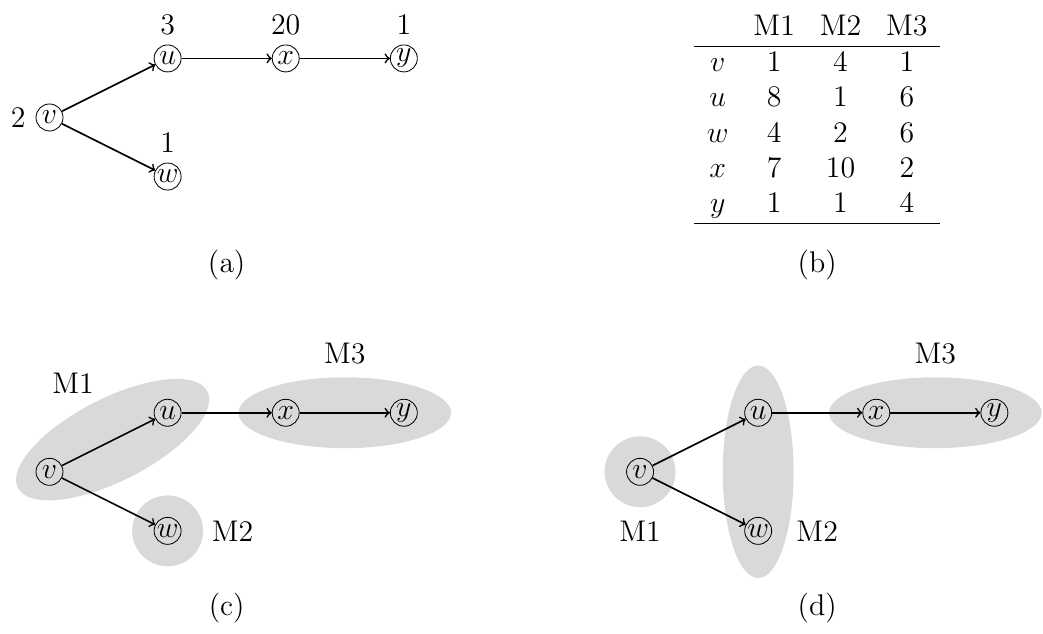}
    \caption{An illustration of the {\mprobshort} problem. (a) A DAG $G$. The number next to each task (vertex) represents the energy consumption required to transfer the task's output. 
    (b) Energy consumption for executing the tasks on three machines: M1, M2, and M3. 
    (c) and (d) are two different assignments of the tasks to the machines. The corresponding energy consumption are respectively $(1+8+2+2+4)+(2+3)=22$ (c) and $(1+1+2+2+4)+(2+3)=15$ (d).}
    \label{fig-DAG-illustraion}
\end{figure}

We also investigate a natural variant of the special case where $k=2$, with the restriction that one of the two machines can execute only a limited number of tasks. For clarity, we reformulate this variant as follows. For a digraph~$G$, two disjoint subsets $V_1, V_2\subseteq V(G)$, and a function $q: V(G)\rightarrow \mathbb{R}_{\geq 0}$, let 
\begin{equation}
\label{eq-bipartition}
    \btranst{(V_1, V_2)}{q}{G}=\sum_{i\in [2]} \sum_{v\in V_i} q(v) \cdot {\bf{1}}_G(v, V_{3-i}).
\end{equation}

\probdef
{{\probva} (\probvashort)}
{A DAG~$G$, two functions $p: V(G)\times [2]\rightarrow \mathbb{R}_{\geq 0}$ and $q: V(G)\rightarrow \mathbb{R}_{\geq 0}$, a rational number~$r$, and an integer $\ell$.}
{Are there disjoint $V_1, V_2\subseteq V(G)$ such that $V_1\cup V_2=V(G)$, $\abs{V_1} \leq \ell$, and  
\[\left(\sum_{v\in V_i, i\in [2]} p(v, i)\right)+\btranst{(V_1, V_2)}{q}{G}\leq r?\]}

\paragraph{Our Main Contributions}
We first establish a complexity dichotomy for {\mprobshort} with respect to the number of machines: the problem is {\nph} if there are at least three machines (Theorem~\ref{thm-DAG-k-NP-hard}), but becomes polynomial-time solvable when there are two machines~(Theorem~\ref{thm-p-k=2}). 
Afterwards, we show that when the input DAG degenerates to a directed path\footnote{Such a DAG is also called a chain in the literature.}, the problem becomes polynomial-time solvable, regardless of the number of machines (Theorem~\ref{thm-P-path}). Additionally, we show that the size-bounded variant {\probvashort}, where there are only two machines, is computationally hard: it is {\wah} with respect to the parameter~$\ell$ (Theorem~\ref{thm-bounded-wah}). As a byproduct of this result, we show that a variant of a minimum cut problem is {\wah} with respect to a natural parameter (Theorem~\ref{thm-sb-minimum-cut-wah}), strengthening its {\nphns} studied in~\cite{DBLP:journals/tcs/ChenSSHY16}.

\section{Related Works}
To the best of our knowledge, the {\mprobshort} problem has not been studied in the literature. However, the problem and its variant defined in the previous section belong to the gigantic family of graph partition problems, which have been extensively and intensively studied in the literature. These problems aim at  dividing either the vertices or the arcs/edges of a given digraph/graph into several sets so that certain structural properties are met or certain optimization objections are achieved (see, e.g.,~\cite{DBLP:journals/siamdm/AlonM11,DBLP:journals/tcs/AndersenBY20,DBLP:journals/tcs/Bang-JensenBHY18,DBLP:journals/siamcomp/BansalFKMNNS14,DBLP:journals/tcs/BodlaenderJ95,DBLP:journals/corr/abs-2307-01109,DBLP:journals/tcs/OumSV14,DBLP:journals/mst/ShachnaiZ17}). 

More specifically, the {\mprobshort} problem falls into the category of task allocation problems~\cite{DBLP:journals/mansci/ErnstJK06,DBLP:journals/tcom/KoprasBIKB22,DBLP:journals/tpds/LiuZGHZS14,DBLP:journals/tpds/PaganiPSCH17}. 
Partitioning a DAG in such problems can achieve various objectives, including minimizing makespan, reducing energy consumption, decreasing communication costs, achieving load balancing, and ensuring fault tolerance. Our model aims to minimize energy consumption. 
In most of the previous works tackling DAG tasks, the cost of data transformation are arc-wisely defined: if a task~$v$ assigned to a machine~$i$ has multiple outneighbors assigned to a different machine~$j$, the output of~$v$ needs to be transferred multiple times from machine~$i$ to machine~$j$, one for each of~$v$'s outneighbors assigned to machine~$j$. 
We build upon previous approaches by considering a more comprehensive energy consumption model for data transformation. We simplify multiple transfers of output data from one machine to another when a task has multiple outneighbors assigned to different machines. This refined modeling approach enables us to evaluate and minimize energy consumption more accurately in DAG-based task allocation.

One noticeable related problem where a similar energy model as ours is adopted is a one studied by Hu~et~al.~\cite{DBLP:conf/infocom/HuBWL19} in 2019. Particularly, this problem takes as the same input as {\probvashort} with~$\ell$ being dropped, and the problem consists in dividing~$V(G)$ into two disjoint sets~$V_1$ and~$V_2$ to minimize $\sum_{v\in V_i, i\in [2]} p(v, i)+\sum_{v\in V_1} q(v)\cdot {\bf{1}}_G(v, V_2)$, under the restriction that there are no arcs from~$V_2$ to~$V_1$. For the problem, Hu~et~al.~\cite{DBLP:conf/infocom/HuBWL19} proposed a polynomial-time algorithm which is, however, pointed out to be flawed by Li~et~al.~\cite{DBLP:journals/tmc/LiLLXJG23}. 

Our studied problem is also related to several specific resource allocation problems. Particularly, when the given DAG does not contain any arc (or $q(v)=0$ for all $v\in V(G)$), {\mprobshort} is equivalent to the problem of maximizing social welfare in multi-agent resource allocation when agents hold $1$-additive utility functions, which is known to be polynomial-time solvable~\cite{DBLP:journals/anor/ChevaleyreEEM08,DBLP:journals/aamas/NguyenNRR14}.\footnote{To see the equivalence, consider each task as a resource, consider each machine as an agent, and consider~$p(v, i)$ as the utility of the resource~$v$ for the agent~$i$.} This special case is also related to a winners allocation problem proposed by Yang~\cite{DBLP:conf/aaai/000121a}, which generalizes the multi-agent resource allocation with two agents, with each holding a $1$-additive utility function.

\section{Problems to Establish Our Results}
Our results are obtained based on the following problems. 

An {\emph{undirected graph}} $G=(V, E)$ is a tuple where~$V$ is a set of vertices and~$E$ is a set of edges.  For simplicity, throughout the remainder of this paper, we will refer to undirected graphs simply as graphs. The set of vertices and the set of edges of~$G$ are also denoted by~$V(G)$ and~$E(G)$, respectively. An {\emph{edge}} between two vertices $v$ and $u$ is denoted by $\edge{v}{u}$. Two vertices are {\it{connected}} in~$G$ if there is path between them.
 For a subset~$E'\subseteq E(G)$, we use $G-E'$ to denote the graph obtained from~$G$ by removing all edges in~$E'$. 
 
For a function $w: S\rightarrow \mathbb{R}$ and a subset $S'\subseteq S$, let $w(S')=\sum_{s\in S'}w(s)$.

\probdef
{Multiway Cut}
{A graph~$G$, a weight function $w: E(G)\rightarrow \mathbb{R}_{\geq 0}$, a set $T=\{t_1, t_2, \dots, t_k\}$ of~$k$ distinct vertices in~$G$ called {\emph{terminals}}, and a number $r$.} 
{Is there a subset~$E'\subseteq E(G)$ such that $w(E')\leq r$, and every two distinct~$t_i, t_j\in T$ are disconnected in $G-E'$?}

Equivalently, {\sc{Multiway Cut}} determines if there is a partition $(V_i)_{i\in [k]}$ of~$V(G)$ such that $t_i\in V_i$ for all $i\in [k]$ and the total weight of edges crossing the partition does not exceed~$r$. 
The {\sc{Multiway Cut}} problem is {\nph} for every $k\geq 3$, but becomes polynomial-time solvable if $k=2$~\cite{DBLP:journals/siamcomp/DahlhausJPSY94}. {\sc{Multiway Cut}} with two terminals is exactly the decision version of the classic problem {\sc{Minimum $s$-$t$-Cut}}. 
The {\sc{Size Bounded Minimum $s$-$t$-Cut}} problem ({\sc{SBM-$s$-$t$-Cut}}) takes as input a graph $G$, a weight function $w: E(G)\rightarrow \mathbb{R}_{\geq 0}$, a pair  $\{s,t\}\subseteq V(G)$ of vertices, and two numbers~$r$ and~$\ell$. The problem asks whether there exists a bipartition $(V_s, V_t)$ of~$V(G)$ such that $\abs{V_s}\leq \ell$, $s\in V_s$, $t\in V_t$, and the total weight of the edges between~$V_s$ and~$V_t$ in~$G$ is at most~$r$. It is known that the {\sc{SBM-$s$-$t$-Cut}} problem is {\nph}~\cite{DBLP:journals/tcs/ChenSSHY16}.

A {\emph{clique}} in a graph is a subset of pairwise adjacent vertices. 

\probdef{Clique}{A graph~$G$ and an integer~$\ell$.}{Does $G$ contain a clique of~$\ell$ vertices?}

{\sc{Clique}} is a well-known {\nph} problem~\cite{DBLP:conf/coco/Karp72}. Moreover, it is {\wah} with respect to~$\ell$ even when restricted to regular graphs~\cite{DBLP:journals/cj/Cai08,DBLP:journals/tcs/Marx06,DBLP:conf/cats/MathiesonS08}.

For a digraph~$G$ and a subset~$A'$ of arcs, $G-A'$ denotes the digraph obtained from~$G$ by removing all arcs in~$A'$. 
For two disjoint subsets $X, Y\subseteq V(G)$, let $\cut{X}{Y}{G}$ be the set of all arcs from~$X$ to~$Y$, i.e., 
$\cut{X}{Y}{G}=\{\arc{v}{u}\in A(G) \setmid v\in X, u\in Y\}$. The decision version of the {\prob{Directed Minimum $s$-$t$ Cut}} problem (DM-$s$-$t$-Cut) is defined as follows. 

\probdef{DM-$s$-$t$-Cut}
{A digraph~$G$, a weight function $w: A(G)\rightarrow \mathbb{R}_{\geq 0}$, a pair $\{s,t\}\subseteq V(G)$ of two vertices, an integer $r$.}
{Is there a bipartition $(V_s, V_t)$ of $V(G)$ such that $s\in V_s$, $t\in V_t$, and $w(\cut{V_s}{V_t}{G})\leq r$?}

Equivalently,  {\sc{DM-$s$-$t$-Cut}} determines if there is a subset of arcs $A'\subseteq A(G)$ such that $w(A')\leq r$ and there is no directed path from~$s$ to~$t$ in $G-A'$. 
We call $w(\cut{V_s}{V_t}{G})$ the {\it{size}} of $\cut{V_s}{V_t}{G}$ with respect to~$w$. 
The {\sc{DM-$s$-$t$-Cut}} problem is polynomial-time solvable~\cite{DBLP:journals/corr/abs-2307-01109}.

\section{Energy-Saving Partition of DAG}
This section presents our results for {\mprobshort}. We first pinpoint the complexity boundary of the problem concerning the number of machines. 

\begin{theorem}
\label{thm-DAG-k-NP-hard}
    {\mprobshort} is {\nph} for every $k\geq 3$.
\end{theorem}

\begin{proof}
We prove Theorem~\ref{thm-DAG-k-NP-hard} by a reduction from the {\sc{Multiway Cut}} problem to the {\mprobshort} problem. Let $(G, w, T, r)$ be an instance of the {\sc{Multiway Cut}} problem, where $T=\{t_1, t_2, \dots, t_k\}$. Without loss of generality, we assume that $k\geq 3$. To construct an instance of the {\mprobshort} problem, we first  arbitrarily fix a linear order on~$V(G)$ and orient edges of~$G$ forwardly. Let~$\overrightarrow{G}$ denote the resulting graph, which is clearly a DAG. It holds that $V(G)=V(\overrightarrow{G})$. Then, we construct a digraph~$G'$ obtained from~$\overrightarrow{G}$ by subdividing all arcs: for each arc~$\arc{v}{u}\in A(\overrightarrow{G})$, we introduce one new vertex~$a(v, u)$, add the arcs $(v, a(v, u))$ and $(a(v, u), u)$, and remove the arc~$\arc{v}{u}$. 
Clearly, each newly introduced vertex~$a(v, u)$ has exactly one outneighbor~$u$ and exactly one inneighbor~$v$. Moreover,~$G'$ remains as a DAG. See Figure~\ref{fig-thm-DAG-k-NP-hard} for an illustration. 
\begin{figure}
    \centering
    \includegraphics[width=\textwidth]{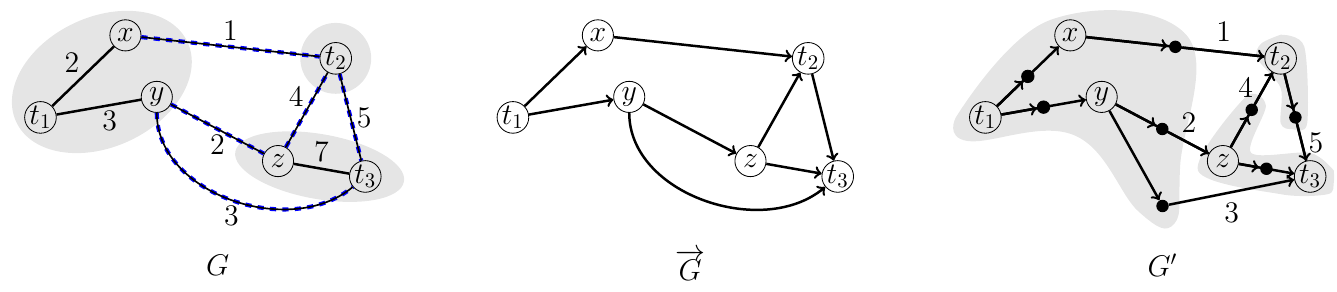}
    \caption{An illustration of the reduction in the proof of Theorem~\ref{thm-DAG-k-NP-hard}. The digraph $\overrightarrow{G}$ is constructed based on the linear ordering $(t_1,x,y,z,t_2,t_3)$. Dark vertices in $G'$ are newly introduced vertices in the reduction.}
    \label{fig-thm-DAG-k-NP-hard}
\end{figure}

Now we construct two functions~$p: V(G')\times [k]\rightarrow \mathbb{R}_{\geq 0}$ and $q: V(G')\rightarrow \mathbb{R}_{\geq 0}$ as follows. 
For each $v\in V(G')$, we define 
\begin{equation}
\label{eq-a}
p(v, i) =
\begin{cases}
+\infty, & \text{if }v=t_j, j\in [k]\setminus \{i\},\\
    0,    & \text{otherwise}. \\
\end{cases}
\end{equation}
For each $v\in V(\overrightarrow{G})$, we define $q(v)=+\infty$, and for each newly introduced vertex $a(v, u)\in V(G')\setminus V(\overrightarrow{G})$ where $\arc{v}{u}\in A(\overrightarrow{G})$, we define $q(a(v, u))=w(\edge{v}{u})$. An instance $(G', p, q,  r)$ of the {\mprobshort} problem is constructed. 

The above reduction clearly runs in polynomial time. 
In the following, we show its correctness. 
For each $v\in V(G)$, let \[B(v)=\{a(v, u) \setmid \arc{v}{u}\in A(\overrightarrow{G})\}.\] For $V'\subseteq V(G)$, let $B(V')=\bigcup_{v\in V'} B(v)$, and let $B[V']=B(V')\cup V'$. 

$(\Rightarrow)$ Assume that the given instance of the {\sc{Multiway Cut}} problem is a {\yesins}; that is, there exists a partition of~$V(G)$ into~$k$ sets~$V_1$,~$V_2$,~$\dots$,~$V_k$ such that $t_i\in V_i$ for all $i\in [k]$, and the total weight of edges crossing the partition is at most~$r$. Note that for every two disjoint $V', V''\subseteq V(G)$, the sets~$B(V')$ and~$B(V'')$ are disjoint. Consequently, $(B[V_i])_{i\in [k]}$ forms a partition of~$V(G')$. Let~$f$ be the assignment function corresponding to this partition, i.e., for every $i\in [k]$ and every $v\in B[V_i]$, we have $f(v)=i$. We claim that~$f$ serves as a {\yes}-witness to the instance of the {\mprobshort} problem constructed above. First, by the definition of $(B[V_i])_{i\in [k]}$ and Equality~\eqref{eq-a}, the following holds: $\sum_{v\in V(G')}p(v, f(v))=0$. Moreover, we have that 
\begin{equation}
\label{eq-b}
\begin{split}
    \transt{f}{q}{G'}& =\sum_{i\in [k]} \sum_{v\in f^{-1}(i)} \sum_{j\in [k]\setminus \{i\}}  q(v) \cdot {\bf{1}}_{G'}(v, f^{-1}(j))\\
    &= \sum_{i\in [k]}\sum_{v\in V_i}\sum_{\substack{a(v, u)\in B(v),\\ j\in [k]\setminus \{i\}, u\in f^{-1}(j)}} q(a(v, u))\\
    &=  \sum_{i\in [k]}\sum_{v\in V_i}\sum_{\substack{\arc{v}{u}\in A(\overrightarrow{G}),\\ j\in [k]\setminus \{i\},  u\in V_j}} q(a(v, u))\\
    & =\sum_{\substack{i,j\in [k], i\neq j,\\ v\in V_i, u\in V_j,\\ \edge{v}{u}\in E(G)}} w(\edge{v}{u})\leq r.\\
\end{split}
\end{equation}
To verify that Equality~\eqref{eq-b} holds, observe that for every $v\in V_i$, where $i\in [k]$, we have $B(v)\subseteq B[V_i]$ and $\outneighbor{G'}{v}=B(v)$. Consequently, only the outputs of the newly introduced vertices $a(v, u)\in V(G')\setminus V(\overrightarrow{G})$ necessitate transfer. 
Additionally, each newly introduced vertex $a(v, u)$ has exactly one outneighbor~$u$. This guarantees the correctness of the transition from the first line to the third line in Equality~\eqref{eq-b}. The transition from the penultimate line to the last line in Equality~\eqref{eq-b} follows from the construction of~$\overrightarrow{G}$ from~$G$ and the definition of the function~$q$.  
Now, we can conclude that the instance of the {\mprobshort} problem is a {\yesins}. 

$(\Leftarrow)$ Assume that the instance of the {\mprobshort} problem is a {\yesins}, i.e., there is an assignment function $f: V(G')\rightarrow [k]$ such that 
\begin{equation}
    \label{eq-c}
\left(\sum_{v\in V(G')} p(v, f(v))\right)+\transt{f}{q}{G'}\leq r.
\end{equation}
By Equality~\eqref{eq-a}, for every terminal~$t_i\in T$, where $i\in [k]$, we have $f(t_i)=i$. Since $p(v, i)=0$ for all~$v \in V(G')\setminus T$ and $i\in [k]$, it follows that $\sum_{v\in V(G')} p(v, f(v))=0$. For each $i\in [k]$, let $V_i=\{v\in V(G) \setmid f(v)=i\}$. We show below that the total weight of edges in~$G$ crossing the partition $(V_i)_{i\in [k]}$ is exactly~$\transt{f}{q}{G'}$, which is at most~$r$. To this end, recall that $q(v)=+\infty$ for every $v\in V(G)$. By Inequality~\eqref{eq-c}, we know that $f(v)=f(v')$ for all  $v\in V(G)$ and $v' \in B(v)$. That is, for every~$v\in V(G)$,~$v$ and all its outneighbors in~$G'$ are assigned the same value by~$f$. To be more precise, for every $i\in [k]$, it holds that $B[V_i]\subseteq f^{-1}(i)$. Since~$B[V_i]$ and~$B[V_j]$ are disjoint whenever~$V_i$ and~$V_j$ are disjoint, this indeed means that $B[V_i] = f^{-1}(i)$ for all $i\in [k]$. Then, by the same reasoning in the proof of the~$(\Rightarrow)$ direction, we infer that Equality~\eqref{eq-b} holds in this direction as well. Therefore, the instance of the {\sc{Multiway Cut}} problem is a {\yesins}. 
\end{proof}

As the optimization version of {\sc{Multiway Cut}} is {\apxh} for every $k\geq 3$~\cite{DBLP:journals/siamcomp/DahlhausJPSY94}, our reduction in the proof of Theorem~\ref{thm-DAG-k-NP-hard} indicates that the optimization version of {\mprobshort} is {\apxh} for every $k\geq 3$.\footnote{The objective of the optimization version of {\mprobshort} is to find an assignment that minimizes the energy consumption.} 
Moreover, as {\sc{Multiway Cut}} remains {\nph} for every $k\geq 3$ even when all edges have the same weight~\cite{DBLP:journals/siamcomp/DahlhausJPSY94}, our reduction also implies that {\mprobshort} remains {\nph} for every $k\geq 3$, even when the two functions~$p$ and~$q$ each has two different values, with one value being identical. Furthermore, astute readers may observe that our reduction can be readily adapted to demonstrate an even more compelling result:

\begin{corollary}
\label{cor-dichotomy-p}
    {\mprobshort} is {\nph} for every $k\geq 3$ even when the two functions~$p$ and~$q$ have overall two different values:~$1$ and~$+\infty$.
\end{corollary} 
The proof of the result can be done by replacing the value~$0$ in Equality~\eqref{eq-a} with~$1$, resetting $q(a(v, u))=1$ for all newly introduced vertices $a(v, u)$, and resetting~$r\coloneqq r+\abs{V(G)}$ in the instance of {\mprobshort} constructed in our reduction. 

When~$p$ is a constant function, {\mprobshort} can be solved trivially. Therefore, we have a complexity dichotomy for {\mprobshort} with respect to the number of different values of the function~$p$. Towards a complexity dichotomy concerning the number of machines, we have the following result. 

\begin{theorem}
\label{thm-p-k=2}
    {\mprobshort} with $k=2$ is linear-time reducible to {\sc{DM-$s$-$t$-Cut}}.
\end{theorem}

\begin{proof}
We derive a linear-time reduction from {\mprobshort} with $k=2$ to {\sc{DM-$s$-$t$-Cut}} as follows. 
Let $(G, p, q, r)$ be an instance of {\mprobshort} where~$G$ is a DAG, and $p: V(G)\times [2]\rightarrow \mathbb{R}_{\geq 0}$ and $q: V(G)\rightarrow \mathbb{R}_{\geq 0}$ are two functions.  We first construct a digraph~$G'$ obtained from~$G$ by performing the following operations:
    \begin{enumerate}
        \item \label{op-1}For each nonsink $v\in V(G)$, we perform the following operations: 
        \begin{enumerate}
            \item \label{op-a} create two vertices~$v^+$ and~$v^-$;
            \item \label{op-b} for every outneighbor~$u$ of~$v$ in~$G$, add the arcs~${\arc{v^+}{u}}$ and~$\arc{u}{v^-}$, and let the weights of both arcs be~$+\infty$;
            \item \label{op-c} add the arcs~$\arc{v}{v^+}$ and~$\arc{v^-}{v}$, and let the weights of both arcs be~$q(v)$;
            \item \label{op-d} remove all arcs from~$v$ to all its outneighbors in~$G$. Therefore,~$V(G)$ forms an independent set in~$G'$. 
        \end{enumerate}
         \item \label{op-2} Create two vertices~$s$ and~$t$, add arcs from~$s$ to all vertices in~$V(G)$, and add arcs from all vertices in~$V(G)$ to~$t$. For each $v\in V(G)$, let the weight of the arc~$\arc{s}{v}$ be~$p(v, 2)$, and let that of~$\arc{v}{t}$ be~$p(v, 1)$. 
    \end{enumerate}
We refer to Figure~\ref{fig-thm-p-k=2} for an illustration of the construction of~$G'$.
\begin{figure}[h]
    \centering
    \includegraphics[width=0.75\textwidth]{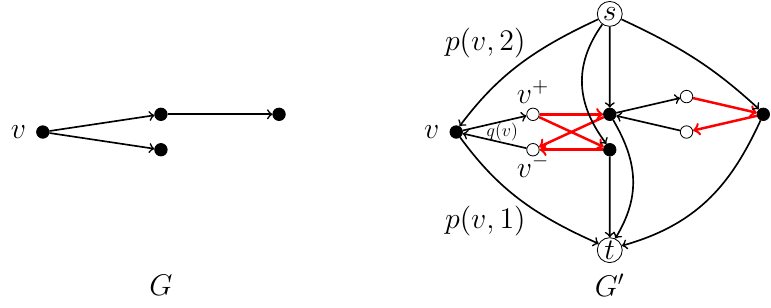}
    \caption{Construction of the digraph~$G'$ from the DAG~$G$ as described in the proof of Theorem~\ref{thm-p-k=2}. Red arcs have positive infinite weights.}
    \label{fig-thm-p-k=2}
\end{figure}
Let $w: A(G')\rightarrow \mathbb{R}_{\geq 0}$ be the function such that for each arc $e\in A(G')$ it holds that~$w(e)$ equals the weight of~$e$ defined above. The instance of {\sc{DM-$s$-$t$-Cut}} is $(G', w, \{s, t\}, r)$. 
The reduction clearly can be carried out in linear time. %$\bigo{n+m}$, where~$n$ and~$m$ are respectively the number of vertices and the number of edges of~$G$. 
It remains to prove its correctness. 

$(\Rightarrow)$ Assume that there are two disjoint $V_1, V_2\subseteq V(G)$ such that 
\begin{equation}
\label{eq-thm-p-k=2-a}
    \left(\sum_{v\in V_i, i\in [2]} p(v, i)\right)+\btranst{(V_1, V_2)}{q}{G}\leq r.
\end{equation} 
Below we construct a partition $(X, Y)$ of $V(G')$ such that $s\in X$, $t\in Y$, and the total weight of edges crossing $(X, Y)$ is at most~$r$. To achieve this, we define the following sets:
\begin{itemize}
    \item For each $i\in [2]$, let $V_i^i=\{v\in V(G) \setmid \outneighborc{G}{v}\subseteq V_i\}$ be the set of vertices~$v\in V(G)$ such that~$v$ and all its outneighbors in~$G$ are in the same set~$V_i$. (Notice that each isolated vertex of~$G$ is either in~$V_1^1$ or in $V_2^2$.)
    \item For each $i\in [2]$, let $V_i^{i\star}=\{v^+ \setmid v\in V_i^i, \outneighbor{G}{v}\neq\emptyset\} \cup \{v^- \setmid v\in V_i^i, \outneighbor{G}{v}\neq\emptyset\}$ be the set of vertices constructed for vertices from~$V_i^i$ by Operation~\eqref{op-a}. 
    \item For each $i, j\in [2]$ such that $i\neq j$, let $V_i^j=\{v\in V_i \setmid \outneighbor{G}{v}\cap V_j\neq\emptyset\}$ be the set of vertices from~$V_i$ having at least one outneighbor from~$V_{j}$ in~$G$. 
    \item For each $i, j\in [2]$ such that $i\neq j$, let $V_i^{j+}=\{v^+ \setmid v\in V_i^j\}$.
    \item For each $i, j\in [2]$ such that $i\neq j$, let $V_i^{j-}=\{v^- \setmid v\in V_i^j\}$.
\end{itemize}
The above defined ten sets are pairwise disjoint, and their union is exactly $V(G')\setminus \{s, t\}$. Noticeably, $V_1=V_1^1\cup V_1^2$ and $V_2=V_2^2\cup V_2^1$. 
Let \[X=V_1^1\cup V_1^{1\star}\cup V_1^2\cup V_1^{2-}\cup V_2^{1-}\cup \{s\},\] 
and let 
    \[Y=V_2^2\cup V_2^{2\star}\cup V_2^1\cup V_1^{2+}\cup V_2^{1+}\cup\{t\}.\] 
It is not difficult to verify that $X\cap Y=\emptyset$ and $X\cup Y=V(G')$. Moreover, it holds that $V_1\subseteq X$ and $V_2\subseteq Y$. 
To conclude this part of the proof, we need to demonstrate that the total weight of edges crossing the bipartition $(X, Y)$ in $G'$ is at most $r$. To achieve this, we identify the edges crossing this bipartition, relying on Claims~\ref{claim-aa}--\ref{claim-fa} presented below.

\begin{claim}
\label{claim-aa}
    None of~$V_1^1\cup V_1^{1\star}\cup V_1^{2-}$ has any outneighbors from~$Y\setminus \{t\}$ in~$G'$. 
\end{claim}

\begin{proof}
   We will prove the claim for the sets~$V_1^1$,~$V_1^{1\star}$,~$V_1^{2-}$ separately. 
   
   Consider first the set $V_1^1$. Let~$v$ be an arbitrary vertex from~$V_1^1$, if any. Apart from~$t$, the outneighbors of $v$ in $G'$ consist of the newly introduced vertex~$v^+$ for~$v$ (if~$v$ is not a sink in $G$), and the newly introduced vertices~$u^-$ for inneighbors~$u$ of~$v$ in~$G$ (if~$v$ is not a source of $G$). Obviously, if~$v^+$ exists, it belongs to the set $V_1^{1\star}$, which is disjoint from~$Y$. Additionally, any~$u^-$ where $\arc{u}{v}\in A(G)$ cannot belong to $V_2^2\cup V_2^1\cup V_1^{2+}\cup V_2^{1+}$. Furthermore,~$u^-$ cannot be in $V_2^{2\star}$ because this would imply $v\in V_2$ contradicting the assumption that $v\in V_1^1\subseteq V_1$. Thus,~$v$ has no outneighbor from $Y\setminus \{t\}$ in~$G'$. Since~$v$ was chosen arbitrarily from $V_1^1$, we conclude that no vertex in~$V_1^1$ has any outneighbors from~$Y\setminus \{t\}$ in~$G'$. 

Next, consider the vertices in the set $V_1^{1\star}$. For every vertex $v^-\in V_1^{1\star}$  where $v\in V_1^1$, the only outneighbor of~$v^{-}$ in~$G'$ is~$v$, which is clearly not in~$Y$.  
For a vertex $v^+\in V_1^{1\star}$  where $v\in V_1^1$, the outneighbors of~$v^+$ in~$G'$ are exactly the outneighbors of~$v$ in the digraph~$G$, all of which belong to~$V_1^1$. As~$V_1^1\subseteq X$, none of the outneighbors of~$v^+$ is from~$Y$.

Finally, consider the set $V_1^{2-}$. By definition,  each vertex $v^-\in V_1^{2-}$ has exactly one outneighbor~$v$ in~$G'$, and~$v$ is from~$V_1^2$. Clearly,~$V_1^2$ and~$Y$ are disjoint.  
\end{proof}

\begin{claim}
\label{claim-ab}
    All arcs from $V_1^2$ to~$Y\setminus \{t\}$ in $G'$ are contained in $\cut{V_1^2}{V_1^{2+}}{G'}$.
\end{claim}

\begin{proof}
The claim is vacuously true if $V_1^2=\emptyset$. Otherwise, let~$v$ be an arbitrary vertex from $V_1^2$. The outneighbors of~$v$ in~$G'$ are from~$\{v^+,t\}\cup \{u^- \setmid \arc{u}{v}\in A(G)\}$. By the definition of~$V_1^2$, it follows that~$v^+$, if exists (when~$v$ is not a sink in~$G$), is contained in $V_1^{2+}$. Now suppose~$v$ has an outneighbor~$u^-$ in~$G'$ such that~$\arc{u}{v}\in A(G)$.  In this case, $u^-$ cannot be from $Y\setminus V_2^{2\star}=V_2^2\cup V_2^1\cup V_1^{2+}\cup V_2^{1+}\cup\{t\}$. If $u^-\in V_2^{2\star}$, then~$u$ must be from~$V_2^2$, which implies that $v\in V_2$, contradicting the fact that $v\in V_1^2\subseteq V_1$. Therefore,~$v$ can only have~$v^+$ and~$t$ as outneighbors in~$G'$. As $v^+\in V_1^{2+}$, the claim holds.
\end{proof} 

\begin{claim}
\label{claim-ac}
    None of~$V_2^2\cup V_2^{2\star}\cup V_2^{1+}$ has any inneighbors from~$X\setminus \{s\}$ in~$G'$.
\end{claim}

\begin{proof}
    We will prove the claim for the sets $V_2^2$, $V_2^{2\star}$, and $V_2^{1+}$ separately. 
    
    Consider first the set $V_2^2$. Let~$v$ be an arbitrary vertex in~$V_2^2$, if any. We consider first the case where~$v$ is neither a sink nor a source of~$G$. By the construction of~$G'$, it holds that $\inneighbor{G'}{v}\setminus \{s\}=\{v^-\}\cup \{u^+ \setmid u\in \inneighbor{G}{v}\}$. Clearly, $v^-\in V_2^{2\star}$. Let~$u$ be a vertex in~$\inneighbor{G}{v}$. If $u\in X$, then~$u$ belongs to~$V_1^2$ because~$u$ has an outneighbor~$v\in V_2^2\subseteq V_2$. It follows that $u^+\in V_1^{2+}$, which is contained in~$Y$. If $u\in Y$, then either~$u^+\in V_2^{2\star}$ (when $u\in V_2^2$) or $u^+\in V_2^{1+}$ (when $u\in V_2^1$). In both cases, we have that~$u^+\in Y$. We can conclude now that none of the inneighbors of~$v$ in~$G'$ is from $X\setminus \{s\}$. In the case where~$v$ is a sink but not a source of~$G$, we have $\inneighbor{G'}{v}\setminus \{s\}=\{u^+ \setmid u\in \inneighbor{G}{v}\}$. If~$v$ is a source but not a sink of~$G$, then $\inneighbor{G'}{v}\setminus \{s\}=\{v^-\}$. If~$v$ is an isolated vertex in~$G$, we have $\inneighbor{G'}{v}\setminus \{s\}=\emptyset$. The proofs for these cases follow the same reasoning as above.

Consider now  a vertex $v^+\in V_2^{2\star}$, where $v\in V_2^2$. By the construction of~$G'$, we know that~$v$ is the only inneighbor of~$v^+$ in~$G'$. For a vertex $v^-\in V_2^{2\star}$, where $v\in V_2^2$, we have $\inneighbor{G'}{v^-}=\outneighbor{G}{v}$. As $v\in V_2^2$, we have $\outneighbor{G}{v}\subseteq V_2=V_2^1\cup V_2^2$. Since $V_2^1\cup V_2^2$ and~$X$ are disjoint, we can conclude that none of the inneighbors of any vertex from~$V_2^{2\star}$ belong to~$X\setminus \{s\}$.

Finally, we prove the claim for~$V_2^{1+}$. Let~$v^+$ be a vertex from~$V_2^{1+}$. By the construction of~$G'$,~$v$ is the only inneighbor of~$v^+$. Obviously, $v\in V_2^1$, which is disjoint from~$X$. 
\end{proof} 

\begin{claim}
\label{claim-ad}
    All arcs from~$V_2^{1-}$ to~$Y$ in~$G'$ are contained in $\cut{V_2^{1-}}{V_2^{1}}{G'}$.
\end{claim}

\begin{proof}
    By the construction of~$G'$, there are no arcs from~$V_2^{1-}$ to $V_1^{2+}\cup \{t\}$ in~$G'$. By Claim~\ref{claim-ac}, there are no arcs from $V_2^{1-}$ to $V_2^2\cup V_2^{2\star}\cup V_2^{1+}$ in~$G'$. This leaves only the possibility that all arcs from~$V_2^{1-}$ to~$Y$ in~$G'$ are contained in $\cut{V_2^{1-}}{V_2^{1}}{G'}$. 
\end{proof}  

The following claim is a consequence of Claims~\ref{claim-aa}--\ref{claim-ad}.

\begin{claim}
\label{claim-fa}
$\cut{X}{Y}{G'}$ is the union of the following pairwise disjoint sets: 
\begin{itemize}
    \item $\cut{\{s\}}{Y}{G'}$, i.e., the set of all arcs from~$\{s\}$ to~$Y$ in the digraph~$G'$. 
    \item $\cut{X}{\{t\}}{G'}$, i.e., the set of all arcs from~$X$ to~$\{t\}$ in the digraph~$G'$. 
    \item $\cut{V_1^2}{V_1^{2+}}{G'}$, i.e., the set of all arcs from $V_1^2$ to $V_1^{2+}$ in the digraph~$G'$.
    \item $\cut{V_2^{1-}}{V_2^{1}}{G'}$, i.e., the set of all arcs from $V_2^{1-}$ to $V_2^1$ in the digraph~$G'$.
\end{itemize}
\end{claim}

%Armed with Claim~\ref{claim-fa}, we examine the size of the cut from~$X$ to~$Y$ in~$G'$. 
By Operation~\eqref{op-2} in the construction of~$G'$, there is an arc from~$s$ to every vertex in $V_2^1\cup V_2^2$, and there is no arc from~$s$ to any vertices from~$Y\setminus (V_2^1\cup V_2^2)$. That is, $\cut{\{s\}}{Y}{G'}$ consists of exactly the arcs from~$s$ to all vertices of~$V_2$ in~$G'$. 
Then, by Operation~\eqref{op-c} in the construction of~$G'$, we have 
\[w(\cut{\{s\}}{Y}{G'})=\sum_{v\in V_2^1\cup V_2^2}w(\arc{s}{v})=\sum_{v\in V_2^1\cup V_2^2}p(v, 2).\]
Analogously, we know that $w(\cut{X}{\{t\}}{G'})=\sum_{v\in V_1^1\cup V_1^2}p(v, 1)$. 
From $V_1=V_1^1\cup V_1^2$ and $V_2=V_2^1\cup V_2^2$, we obtain 
\begin{equation}
\label{eq-z}
w(\cut{\{s\}}{Y}{G'})+w(\cut{X}{\{t\}}{G'})=\sum_{v\in V_i, i\in [2]} p(v, i).
\end{equation}
Now we analyze  $w(\cut{V_1^2}{V_1^{2+}}{G'})$. As for every $v^+\in V(G')$ where~$v\in V(G)$,~$v$ is the only inneighbor of~$v^+$ in~$G'$, we have that \[\cut{V_1^2}{V_1^{2+}}{G'}=\{\arc{v}{v^+} \setmid v\in V_1^2\}.\] 
It follows that 
\[w(\cut{V_1^2}{V_1^{2+}}{G'})=\sum_{v\in V_1^2} w(\arc{v}{v^+})=\sum_{v\in V_1^2} q(v).\]
Analogously, we can obtain that $w(\cut{V_2^{1-}}{V_2^{1}}{G'})=\sum_{v\in V_2^1} q(v)$. 

By Claim~\ref{claim-aa}, and the facts that $V_1=V_1^1\cup V_1^2$,  $V_2\subseteq Y$, and every vertex from~$V_1^2$ has at least one outneighbor from~$V_2$ in~$G$, we have that $\sum_{v\in V_1} q(v) \cdot {\bf{1}}_{G}(v, V_2)=\sum_{v\in V_1^2} q(v)$. Similarly, it holds that $\sum_{v\in V_2} q(v) \cdot {\bf{1}}_{G}(v, V_1)=\sum_{v\in V_2^1} q(v)$. 
We arrive at 
\begin{equation}
\label{eq-zz}
w(\cut{V_1^2}{V_1^{2+}}{G'})+w(\cut{V_2^{1-}}{V_2^{1}}{G'})=\btranst{(V_1, V_2)}{q}{G}.
\end{equation}
From Claim~\ref{claim-fa}, and Equalities~\eqref{eq-z} and~\eqref{eq-zz}, we have that 
\[w(\cut{X}{Y}{G'})=\left(\sum_{v\in V_i, i\in [2]} p(v, i)\right)+\btranst{(V_1, V_2)}{q}{G},\] 
which is at most~$r$ by Inequality~\eqref{eq-thm-p-k=2-a}. Therefore, the instance of the {\sc{DM-$s$-$t$-Cut}} problem is a {\yesins}.  

$(\Leftarrow)$ Assume that there is a bipartition $(X, Y)$ of~$V(G')$ such that $s\in X$, $t\in Y$, and $w(\cut{X}{Y}{G'})\leq r$. Without loss of generality, we  assume that~$w(\cut{X}{Y}{G'})$ is minimized among all bipartitions of~$V(G')$. We show that $\cut{X}{Y}{G'}$ possesses several structural properties which help us construct a desired assignment of~$V(G)$. 

The first property, as formally articulated in the ensuing claim, affirms that when a vertex~$v$ in~$G$ is included in the same part of $(X, Y)$ alongside all its outneighbors in~$G$, retaining both newly introduced vertices for $v$ within the same part as $v$ emerges as an optimal strategy.

\begin{claim}
\label{claim-a}
    Let $v\in V(G)$ such that $\outneighbor{G}{v}\neq\emptyset$ and all outneighbors of~$v$ in~$G$ are in the same $Z\in \{X, Y\}$. Let $\overline{Z}=V(G')\setminus Z$. Let $X'=Z\cup \{v^+, v^{-}\}$ and let $Y'=\overline{Z}\setminus \{v^+, v^{-}\}$. Then, it holds that $w(\cut{X'}{Y'}{G'})\leq w(\cut{X}{Y}{G'})$ if $s\in X'$, and  $w(\cut{Y'}{X'}{G'})\leq w(\cut{X}{Y}{G'})$ if $s\in Y'$. 
\end{claim}

\begin{proof}
    Let~$v$ be as stipulated in Claim~\ref{claim-a}. We consider first the case where $\outneighborc{G}{v}\subseteq X$ (i.e., $Z=X$). In this case, $X'=X\cup \{v^+, v^-\}$ and $s\in X'$. We need to prove that $w(\cut{X'}{Y'}{G'})\leq w(\cut{X}{Y}{G'})$. To this end, observe that  none of the arcs of $\{\arc{u}{v^-} \setmid u\in \outneighbor{G}{v}\}$ is in~$\cut{X}{Y}{G'}$, due to the infinite weight of these arcs (see~\eqref{op-b}). It follows that~$v^-\in X$. If $v^+\in X$, then $(X',Y')=(X,Y)$, and thus $w(\cut{X'}{Y'}{G'})= w(\cut{X}{Y}{G'})$; we are done. If $v^+\in Y$, then since $\inneighbor{G'}{v^+}=\{v\}$, we have that $w(\cut{X}{Y}{G'})-w(\cut{X'}{Y'}{G'})=w(\arc{v}{v^+})\geq 0$. 

We consider now the case where $\outneighborc{G}{v}\subseteq Y$ (i.e., $Z=Y$). In this case, $X'=Y\cup \{v^+, v^{-}\}$, $Y'=X\setminus \{v^+, v^-\}$, and $s\in Y'$. Observe that  none of the arcs in $\{\arc{v^+}{u} \setmid u\in \outneighbor{G}{v}\}$ belongs to~$\cut{X}{Y}{G'}$, due to their infinite weight (see~\eqref{op-b}). It follows that $v^+\in Y$. If $v^-\in Y$, then $(Y',X')=(X,Y)$, and thus $w(\cut{Y'}{X'}{G'})= w(\cut{X}{Y}{G'})$; we are done. If $v^-\in X$, then since $\outneighbor{G'}{v^-}=\{v\}$, we have that $w(\cut{X}{Y}{G'})-w(\cut{Y'}{X'}{G'})=w(\arc{v^-}{v})\geq 0$. 
\end{proof} 

By Claim~\ref{claim-a}, we may assume that, for every $v\in V(G)$ such that $\outneighbor{G}{v}\neq\emptyset$, if~$\outneighborc{G}{v}$ are contained in the same $Z\in\{X, Y\}$, then $\{v^+, v^-\}$ are also contained in~$Z$. 

The subsequent property essentially posits that for any vertex~$v\in V(G)$ placed in part~$X$, if at least one of~$v$'s outneighbors in~$G$ is placed in the opposite part~$Y$, then the arc $\arc{v}{v^+}$, whose weight equals the amount of energy consumption for transferring the output of~$v$, must be in $\cut{X}{Y}{G'}$.

\begin{claim}
    \label{obs-b}
    Let $v\in X\cap V(G)$ such that $\outneighbor{G}{v}\cap Y \neq \emptyset$. Then, $\arc{v}{v^+}\in \cut{X}{Y}{G'}$ and none of the arcs entering or leaving~$v^-$ is contained in~$\cut{X}{Y}{G'}$.
\end{claim}

\begin{proof}
    Let~$v$ be as stipulated in the claim. Let~$u$ be an arbitrary vertex from~$\outneighbor{G}{v}\cap Y$. It must hold that $v^+\in Y$, because if not, the arc~$\arc{v^+}{u}$, which has an infinite weight, would be included in~$\cut{X}{Y}{G'}$, contradicting with $w(\cut{X}{Y}{G'})\leq r$. Since~$v$ is the only outneighbor of~$v^-$ in~$G'$ and~$v\in X$, the only arc~$\arc{v^-}{v}$ leaving~$v^-$ is excluded from~$\cut{X}{Y}{G'}$. Finally, recall that $\inneighbor{G'}{v^-}=\outneighbor{G}{v}$, and by the definition of the function~$w$, every arc $\arc{u}{v^-}$ where $u\in \outneighbor{G}{v}$ has an infinite weight. Since $w(\cut{X}{Y}{G'})\leq r$, none of the arcs entering~$v^-$ can be in~$\cut{X}{Y}{G'}$. 
\end{proof} 

By Claim~\ref{obs-b}, for every $v\in X\cap V(G)$ having at least one outneighbor in~$G$ that is put in~$Y$, we may assume that~$v^+\in Y$. 
Moreover, as~$v$ is the only outneighbor of~$v^-$ in~$G'$, we may assume that~$v^-\in X$.

Applying analogous reasoning, we derive the following assertion.

\begin{claim}
\label{obs-c}
Let $v\in Y\cap V(G)$ such that $\outneighbor{G}{v}\cap X \neq\emptyset$. Then, $\arc{v^-}{v}\in \cut{X}{Y}{G'}$, and none of the arcs entering or leaving~$v^+$ is contained in~$\cut{X}{Y}{G'}$.
\end{claim}

For every $v\in Y\cap V(G)$ as in Claim~\ref{obs-c}, we may assume that $v^-\in X$ and $v^+\in Y$. 

Let $V_1=V(G)\cap X$ and let $V_2=V(G)\cap Y$. We show below that the assignment corresponding to $(V_1, V_2)$ is a {\yes}-witness of the instance of {\mprobshort}. 
By the definition of the weight function~$w$, we have that 
\begin{equation}
\label{eq-x}
    \sum_{v\in V_1} p(v, 1)+\sum_{v\in V_2}p(v, 2) = \sum_{v\in V_1} w(\arc{v}{t})+\sum_{v\in V_2} w(\arc{s}{v}).
\end{equation}
Let $V_1'=\{v\in V_1 \setmid \outneighbor{G}{v}\cap Y\neq\emptyset\}$ 
and let~$V_2=\{v\in V_2\setmid \outneighbor{G}{v}\cap X\neq\emptyset\}$. By Claims~\ref{claim-a}--\ref{obs-c}, we have: 
\begin{itemize}
    \item For every $i\in [2]$ and every $v\in V_i\setminus V_i'$, $\{v^+,v^-\}$ are contained in the same part of $(X, Y)$ as~$v$, whenever~$v$ is not a sink in~$G$.
    \item For every $v\in V_1'\cup V_2'$, $v^+\in Y$ and $v^-\in X$.
\end{itemize}
The following equalities follow:
\begin{equation}
\label{eq-xx}
   \sum_{v\in V_1} q(v)\cdot {\bf{1}}_G(v, V_2)=\sum_{v\in V_1'} w(\arc{v}{v^+}),
\end{equation}
\begin{equation}
\label{eq-xxx}
   \sum_{v\in V_2} q(v) \cdot {\bf{1}}_G(v, V_1)=\sum_{v\in V_2'} w(\arc{v^-}{v}).
\end{equation}
The sum of the left sides of Equalities~\eqref{eq-x}--\eqref{eq-xxx} is $\left(\sum_{v\in V_i, i\in [2]} p(v, i)\right)+\btranst{(V_1, V_2 )}{q}{G}$, and the sum of the right sides of them is $w(\cut{X}{Y}{G'})$. From $w(\cut{X}{Y}{G'})\leq r$, we know that the instance of the {\mprobshort} problem is a {\yesins}. 
\end{proof}

As the {\sc{DM-$s$-$t$-Cut}} problem can be solved in time $\bigo{n\cdot m \cdot \log (n^2/m)}$~\cite{HAO1994424}, by Theorem~\ref{thm-p-k=2}, we reach the following result.

\begin{corollary}
{\mprobshort} with $k=2$ can be solved in time $\bigo{n\cdot m \cdot \log (n^2/m)}$, where~$n$ and~$m$ are respectively the number of vertices and the number of edges of the input digraph.
\end{corollary}

Next, we derive a polynomial-time algorithm for a special case of {\mprobshort} where the input DAG is a directed path. Tasks having such precedence dependencies are relevant to many applications (see, e.g.,~\cite{DBLP:conf/europar/AbaZM17,DBLP:journals/ijfcs/JansenS10}). 

\begin{theorem}
\label{thm-P-path}
   The {\mprobshort} problem can be solved in time $\bigo{n\cdot k^3 \cdot M}$ when the given DAG is a directed path with~$n$ vertices, where~$M$ represents the number of bits to encode the largest value of the functions~$p$ and~$q$. 
\end{theorem}

\begin{proof}
Let $I=(G, p, q, r)$ be an instance of the {\mprobshort} problem, where~$G$ is a directed path, $p: V(G)\times [k]\rightarrow \mathbb{R}_{\geq 0}$ and $q: V(G)\rightarrow \mathbb{R}_{\geq 0}$ are two functions, and~$r$ is a number. Let~$M$ be the number of bits to encode the largest value of the functions~$p$ and~$q$. We derive a dynamic programming algorithm to solve the problem as follows. 

    Let $(v_1, v_2, \dots, v_n)$ be the directed path representing~$G$. 
    We maintain a table $S(i, j)$, where $i\in [n]$ and $j\in [k]$. Specifically,~$S(i, j)$ is defined as the value of an optimal solution to the instance~$I$ restricted to the first~$i$ vertices, under the condition that~$v_i$ is assigned to machine~$j$:
    \[S(i,j)=p(v_i, j)+\min_{\substack{f: \{v_1, \dots, v_{i}\}\rightarrow [k],\\ f(v_i)=j}} \left\{\sum_{x\in [i-1]} p(v_{x}, f(v_x))+\sum_{\substack{x\in [i-1]~{\text{such that}}\\ f(v_x)\neq f(v_{x+1})}} q(v_x)\right\}.\]
    By definition, for each $j\in [k]$, we have $S(1,j)=p(v_1, j)$. 
    We use the following recursion to update the table:
    \[S(i,j)=p(v_i, j)+\min\left\{S(i-1, j), \min_{j'\in [k]\setminus \{j\}}S(i-1, j')+q(v_{i-1})\right\}.\]
    After computing the table, we conclude that~$I$ is a {\yesins} if and only if \[\min_{j\in [k]}S(n,j)\leq r.\] 

As there are at most $n\cdot k$ entries in the table, and computing each entry requires checking up to~$k$ other entries and performing addition and comparison operations~$\bigo{k}$ times on numbers that can be encoded within~$M$ bits, the entire table can be computed in time $\bigo{n\cdot k^3\cdot M}$.
\end{proof}

\section{Sized-Bounded Energy-Saving Partition of DAG with Two Machines}
 
Now we switch our focus to the natural variant {\probvashort}, where there are only two machines and one of them is capable of executing only a limited number~$\ell$ of tasks. 
  
Recall that the {\sc{SBM-$s$-$t$-Cut}} problem is exactly {\sc{Multiway Cut}} where $k=2$ with an additional restriction that one of the parts in the desired bipartition contains at most~$\ell$ vertices. 
However, despite the {\nphns} the {\sc{SBM-$s$-$t$-Cut}} problem, the reduction in the proof of Theorem~\ref{thm-DAG-k-NP-hard} is insufficient to show the {\nphns} of the {\probvashort} problem. The reason is that, according to the reduction, for a bipartiton $(V_s, V_t)$ of the graph~$G$ in an instance of the  {\sc{SBM-$s$-$t$-Cut}} problem, in the corresponding bipartition $(V_1, V_2)$ of the created DAG $G'$, $V_1$ needs to contain $V_s$ and all newly introduced vertices for vertices in $V_s$. However, the number of these newly introduced vertices cannot be bounded in an ``exact'' manner, which hinders the setup of the value of~$\ell$. We thus need to modify the reduction in order to pinpoint a countable bound of $\abs{V_1}$ and thus the value of $\ell$. To this end, through a reduction from the {\sc{Clique}} problem restricted to regular graphs, we  show that the {\sc{SBM-$s$-$t$-Cut}} problem is {\nph} even in a special case where the degrees of the vertices in the input graph is nearly regular. As a matter of fact, we show that the problem is {\wah} with respect to~$\ell$. Despite the widespread attention this problem has received, it is remarkable that such a {\wahns} result has not been previously documented. Therefore, our reduction bolsters the {\nphns} of {\sc{SBM-$s$-$t$-Cut}} as examined in \cite{DBLP:journals/tcs/ChenSSHY16}. We believe this {\wahns} result is of independent interest, and hence explitecty state it in the following theroem.
%This observation prompts an exploration into whether the natural parameter~$\ell$ might offer a potential tractability insight within the realm of parameterized complexity theory. 

A {\it{parameterized problem}} is a subset of $\Sigma^*\times \mathbb{N}$, where~$\Sigma$ is a finite alphabet. A parameterized problem can be either {\it{fixed-parameter tractable}} ({\fpt}) or {\it{fixed-parameter intractable}}.
In particular, a parameterized problem is {\fpt} if there is an algorithm which correctly determines for each instance $(I, \kappa)$ of the problem whether 
$(I, \kappa)$ is a {\yesins} in time $\bigo{f(\kappa)\cdot \abs{I}^{\bigo{1}}}$, where~$f$ is a computable function and~$\abs{I}$ is the size of~$I$.
{\wah} problems are considered intractable, in the sense that unless {\fpt}$ = ${\wa} (which is widely believed to be unlikely), they do not admit any {\fpt}-algorithms. For greater details on parameterized complexity theory, we refer to~\cite{DBLP:books/sp/CyganFKLMPPS15,DBLP:series/txcs/DowneyF13,DBLP:journals/siamcomp/DowneyF95}. 

\begin{theorem}
\label{thm-sb-minimum-cut-wah}
The {\sc{SBM-$s$-$t$-Cut}} problem is {\wah} with respect to~$\ell$. Furthermore, this holds even when the edges of the given digraph have at most two distinct weights, and all vertices, except~$s$ and~$t$, have the same degree. 
\end{theorem}

\begin{proof}
    Let $(G, \ell)$ be an instance of the {\sc{Clique}} problem, where~$G$ is a regular graph with~$n$ vertices. Let~$d$ denote the degree of the vertices in~$G$. 
    We assume that $d>\ell$, as otherwise the problem becomes trivial to solve. 
    Let~$G'$ be the graph obtained from~$G$ by adding two new vertices,~$s$ and~$t$, both of which are adjacent to all vertices of~$G$. Define a weight function $w: E(G')\rightarrow \mathbb{R}_{\geq 0}$ such that, for every edge $\edge{x}{y}\in E(G')$, the following holds:  
    \begin{equation*}
    w(\edge{x}{y})=
    \begin{cases}
        d+2, & \text{if }s\in \{x,y\},\\
        1,& \text{otherwise}.
    \end{cases}   
    \end{equation*} 
Let $r=n\cdot (d+2)-\ell^2$. Let $I=(G', w, \{s, t\}, r, \ell+1)$, which represents an instance of the {\sc{SBM-$s$-$t$-Cut}} problem. Notice that both~$s$ and~$t$ have the same degree $n$, and every other vertex has a degree of $d+1$ in~$G'$. 
We prove the correctness of the reduction as follows.
   
   $(\Rightarrow)$ Assume that~$G$ contains a clique~$K\subseteq V(G)$ of~$\ell$ vertices. Let $V_s=K\cup \{s\}$ and let $V_t=\{t\}\cup V(G)\setminus K$. Obviously, $(V_s, V_t)$ forms a bipartition of~$V(G')$. By the definition of the function~$w$, the total weight of edges crossing $(V_s, V_t)$ is $(n-\ell)\cdot (d+2)+\ell\cdot (d-\ell+1)+\ell=n\cdot (d+2)-\ell^2=r$. Therefore, $(V_s,V_t)$ is a {\yes}-witness for the instance~$I$ of the {\sc{SBM-$s$-$t$-Cut}} problem.

   $(\Leftarrow)$ Assume that~$I$ is a {\yesins} of the {\sc{SBM-$s$-$t$-Cut}} problem, i.e., there exists a bipartition $(V_s, V_t)$ of~$V(G')$ such that $s\in V_s$, $\abs{V_s}\leq \ell+1$, $t\in V_t$, and the total weight of the edges between~$V_s$ and~$V_t$ is at most~$r$ with respect to the weight function~$w$. Let~$v$ be an arbitrary vertex in $V_t\setminus \{t\}$, if any. Observe that as the weight of the edge between~$s$ and~$v$ is $d+2$, and each of the other $d+1$ edges incident to~$v$ in~$G'$ has a weight of~$1$, moving~$v$ from~$V_t$ to~$V_s$ strictly reduces the total weight of the edges between~$V_s$ and~$V_t$. Based on this observation, we may assume that $\abs{V_s}=\ell+1$. Let $K=V_s\setminus \{s\}$, so $\abs{K}=\ell$. The total weight of the edges between~$\{s\}$ and $V(G)\setminus K$ is $(n-\ell)\cdot (d+2)$. The total weight of the edges between~$\{t\}$ and~$K$ is~$\ell$. This implies that the total weight of the edges between~$K$ and $V(G)\setminus K$ is at most $r-(n-\ell)\cdot (d+2)-\ell=\ell\cdot (d-\ell+1)$. Since each of these edges has a weight of~$1$, there are at most $\ell\cdot (d-\ell+1)$ edges between~$K$ and $V(G)\setminus K$ in~$G$. Given that~$G$ is $d$-regular, it follows that there are at least $\frac{\abs{K}\cdot d-\ell\cdot (d-\ell+1)}{2}=\frac{\ell\cdot (\ell-1)}{2}$ edges within~$K$. This is possible only if~$K$ is a clique in~$G$. Therefore, we can conclude that $(G, \ell)$ is a {\yesins} of the {\sc{Clique}} problem.
\end{proof}

We are ready to show the {\nphns} of the {\probvashort} problem. 

\begin{theorem}
\label{thm-bounded-wah}
The {\probvashort} problem is {\nph}. Furthermore, this holds even when the function~$p$ has two distinct values and the function $q$ has three distinct values.
\end{theorem}

\begin{proof}
    We prove the theorem by adapting the reduction presented in the proof of Theorem~\ref{thm-DAG-k-NP-hard}. Let $I=(G, w, \{s, t\}, r, \ell)$ be an instance of the {\sc{SBM-$s$-$t$-Cut}} problem. Note that, without $\ell$,~$I$ is an instance of the {\sc{Multiway Cut}} problem, where~$s$ and~$t$ are two terminals. Following the proof of Theorem~\ref{thm-sb-minimum-cut-wah}, we assume that every vertex except~$s$ and~$t$ has the same degree~$d$ in $G$, and that either of~$s$ and~$t$ is adjacent to all vertices from $V(G)\setminus \{s,t\}$. Let $n=\abs{V(G)\setminus \{s,t\}}$, so both~$s$ and~$t$ have degree~$n$ in~$G$. 
    
    We follow the same procedure as in the proof of Theorem~\ref{thm-DAG-k-NP-hard} to construct a DAG $\overrightarrow{G}$ from~$G$, with the additional restriction that~$s$ and~$t$ are the first and last vertices, respectively, in the ordering of~$V(G)$ used to guide the construction of $\overrightarrow{G}$. Consequently,~$s$ and~$t$ are the source and the sink of~$\overrightarrow{G}$, respectively. We then use the same procedure as in the proof of Theorem~\ref{thm-DAG-k-NP-hard}  to construct a DAG~$G'$ from~$\overrightarrow{G}$. Next, we obtain a DAG~$G''$ from~$G'$ by adding the following vertices and arcs for each vertex $v\in V(G)\setminus \{s,t\}$:
    \begin{enumerate}
        \item[(1)] Create a set $J(v)$ of~$i$ new vertices, where~$i$ is the number of inneighbors of~$v$ in~$\overrightarrow{G}$. 
        \item[(2)] Add arcs from~$v$ to all vertices in~$J(v)$.
    \end{enumerate} 
    We refer to Figure~\ref{fig-thm-bounded-wah} for an illustration.

    \begin{figure}
        \centering
        \includegraphics[width=\textwidth]{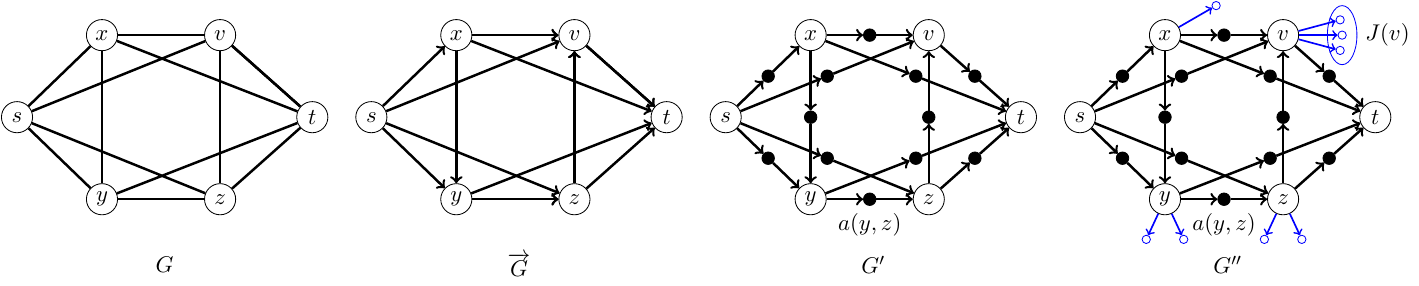}
        \caption{An illustration of the reduction in the proof of Theorem~\ref{thm-bounded-wah}. The DAG $\overrightarrow{G}$ is constructed based on the ordering $(s,x,y,z,v,t)$ on $V(G)$.}
        \label{fig-thm-bounded-wah}
    \end{figure}
    
    We now define the functions~$p$ and~$q$. The values of these functions for vertices from~$V(G')$ remain the same as in the proof of Theorem~\ref{thm-DAG-k-NP-hard}, where~$s$ is treated as~$t_1$ and~$t$ as~$t_2$. %Therefore, $p(s,2)=p(t,1)=+\infty$. 
    Therefore, transferring the outputs of vertices in $V(\overrightarrow{G})$, including~$s$ and~$t$, requires an infinite amount of energy. Specifically, computing 
$s$ on machine 2 and $t$ on machine 1 also consumes infinite energy (i.e., $p(s,2)=p(t,1)=+\infty$). However, for any other computation of tasks on either machine, the energy consumption is zero. Additionally, for every arc $\arc{v}{u}$ in $\overrightarrow{G}$, transferring the output of the corresponding vertex $a(v,u)$, created during the reduction, consumes 
$w(\{v,u\})$ units of energy.   

    For all vertices $a\in V(G'')\setminus V(G')$, we set $q(a)=p(a,1)=p(a,2)=0$. %This means that the computation time for tasks corresponding to~$J(v)$ for all $v\in V(G)\setminus \{s,t\}$ are negligible on both machines. 
    Let $\ell'=\ell+n+d\cdot (\ell-1)$. Define $g(I)=(G'', p, q, r, \ell')$, which represents an instance of the {\probvashort} problem. The correctness of the reduction is proved as follow.
    
%For every vertex $v\in V(G)\setminus \{s,t\}$, we define $\widetilde{J}(v)=J(v)\cup \{a(v,u)\setmid \arc{v}{u}\in A(G)\}$ as the set of all the vertices created for $v$ during the reduction. Note that $\widetilde{J}(v)$ is exactly the set of outneighbors of~$v$ in~$G''$.
    
    $(\Rightarrow)$ Assume that there is a bipartition $(V_s,V_t)$ of $V(G)$ such that $s\in V_s$, $\abs{V_s}\leq \ell$, $t\in V_t$, and the total weight of edges crossing $(V_s, V_t)$ is at most $r$ with respect to $w$. 
    %Let~$B(V_s)=\{a(v, u) \setmid v\in V_s, \arc{v}{u}\in A(\overrightarrow{G}\}$ and let $B(V_t)=\{a(v, u) \setmid v\in V_t, \arc{v}{u}\in A(\overrightarrow{G}\}$. In addition, l
    Let 
    \[B[V_s]=V_s\cup \{a(v, u) \setmid v\in V_s, \arc{v}{u}\in A(\overrightarrow{G}\}\cup \left(\bigcup_{v\in V_s\setminus \{s\}}J(v)\right)\]
    be the set that consists of $V_s$ and all vertices created for vertices in $V_s$ during the reduction.   
    Similarly, define
    \[B[V_t]=V_t\cup \{a(v, u) \setmid v\in V_t, \arc{v}{u}\in A(\overrightarrow{G}\}\cup \left(\bigcup_{v\in V_t\setminus \{t\}}J(v)\right),\]
    which is the set consisting of $V_t$ and all vertices created for vertices in $V_t$ in the reduction. 
    Following the reasoning principle used in the proof of Theorem~\ref{thm-DAG-k-NP-hard}, it can be verify that $(B[V_s], B[V_t])$ is a partition of $V(G'')$, and the total energy consumption corresponding to this partition is at most~$r$ with respect to the functions~$p$ and~$q$. Now we analyze the size of ${B[V_s]}$. By the construction of~$G''$, we have that 
    \[\babs{\{a(v, u) \setmid v\in V_s, \arc{v}{u}\in A(\overrightarrow{G}\}}=\sum_{v\in V_s} \babs{\outneighbor{v}{\overrightarrow{G}}},\] 
    and 
    \[\babs{\bigcup_{v\in V_s\setminus \{s\}}J(v)}=\sum_{v\in V_s\setminus \{s\}} \babs{\inneighbor{v}{\overrightarrow{G}}}.\]
    
    It follows that 
   \[\abs{B[V_s]}= \abs{V_s}+\abs{\outneighbor{s}{\overrightarrow{G}}}+\sum_{v\in V_s\setminus \{s\}}(\abs{\inneighbor{v}{\overrightarrow{G}}}+\abs{\outneighbor{v}{\overrightarrow{G}}}) \leq \ell+n+d\cdot (\ell-1)=\ell'\]
We can now conclude that $(B[V_s], B[V_t])$ is {\yes}-witness for the instance~$g(I)$ of the {{\probvashort}} problem. 
    
    $(\Leftarrow)$ Assume that there is a bipartition $(V_1, V_2)$ of $V(G'')$ such that  
    \[\sum_{v\in V_i}p(v,i)+\btranst{(V_1, V_2)}{q}{G''}\leq r.\] 
    Recall that the value of the function $p$ is either zero or $+\infty$. Hence, it holds that $\btranst{(V_1, V_2)}{q}{G''}\leq r$. 
    In addition, due to the infinity of $p(s,2)$ and $p(t,1)$, we know that $s\in V_1$ and $t\in V_2$. Moreover, as for every vertex $v\in V(G)$ it holds that $q(v)=\infty$, $v$  are in the same part of $(V_1, V_2)$ with all its outneighbors in $G''$ (i.e., vertices from $\{a(v, u) \setmid v\in V_s, \arc{v}{u}\in A(\overrightarrow{G}\}\cup J(v)$. Let $V_s=V_1\cap V(G)$, and let $V_t=V_2\cap V(G)$. We have that $s\in V_s$ and $t\in V_t$. Let $\{v,u\}$ be any arbitrary edge in~$G$ between~$V_1$ and~$V_2$. Without loss of generality, assume that $v\in V_1$ and $u\in V_2$. Then, either $\arc{v}{u}$ or $\arc{u}{v}$ is an arc in $\overrightarrow{G}$. In the former case, there is a vertex $a(v, u)\in V_1$ and an arc from $a(v, u)$  to~$u$ in $G''$. In the latter case there is a vertex $a(u,v)\in V_2$ and an arc from $a(u, v)$ to $v$ in $G''$. Importantly, we have that $q(a(v, u))=w(\{v,u\})$ in the former case, and have $q(a(u, v))=w(\{v,u\})$ in the latter case. Recall now that the transfer of the output of all nonsink vertices in $G''$ except those in $\{a(v,u)\setmid \arc{v}{u}\in A(\overrightarrow{G})\}$ requires infinite energy. As a result, 
    \begin{equation*}
    \begin{split}
         \btranst{(V_1, V_2)}{q}{G''}= & \sum_{v\in V_1, u\in V_2,\arc{v}{u}\in A(\overrightarrow{G})}q(a(v,u))+\sum_{v\in V_1, u\in V_2,\arc{u}{v}\in A(\overrightarrow{G})}q(a(u,v))\\
                                     = & \sum_{v\in V_1, u\in V_2, \edge{v}{u}\in E(G)}w(\{v,u\})\leq r. \\
    \end{split}   
    \end{equation*}
    We can conclude now that the instance of the {\sc{SBM-$s$-$t$-Cut}} problem is a {\yesins}.

    To see that the  {\nphns} of the  {\sc{SBM-$s$-$t$-Cut}} holds even in the special case stipulated in the theorem, note that in the reduction, the function~$p$ takes only two distinct values. Furthermore, based on the reduction above and Theorem~\ref{thm-sb-minimum-cut-wah}, the reduction remains valid even when the values of the $q$ function for all vertices in $G'$ are limited to two possibilities. Combined with the zero value of the function for vertices in $V(G'')\setminus V(G')$, the function $q$ can have at most three distinct values.
\end{proof}

\section{Conclusion}
We have studied the complexity of the {\mprobshort} problem and its sized-bounded variant {\probvashort}. Our study offers several dichotomy results on the complexity of the {\mprobshort} problem, concerning the number~$k$ of machines and the number of different values of the function~$p$. However, an intriguing open question persists: can {\mprobshort} be solved in polynomial time when the function~$q$ is constant? To gain deeper insight into this problem, consider a more restricted case where $k=3$ and $q$ is constantly equal to $1$. This special case of the problem seeks to partition vertices in a graph into three parts, with the goal of minimizing the corresponding energy consumption for completing the tasks plus the number of vertices that have at least one outneighbor in a different part. In addition to this open problem, the unresolved inquiry into whether {\probvashort} retains its status as {\nph} when~$q$ exhibits only two distinct values adds another layer of complexity to this fascinating problem.

\section*{Acknowledgements}
The authors would like to express their gratitude to the three anonymous reviewers of Theoretical Computer Science for their constructive feedback and insightful suggestions, which greatly contributed to the improvement of this paper

\end{document}